\documentclass[11pt]{article}  
\usepackage{fullpage}
\usepackage{enumitem}

\def\draft{0}

\usepackage{amsfonts}
\usepackage{xcolor}
\usepackage{amsfonts,amsthm,amsmath,amssymb}
\usepackage{tikz}
\usepackage[normalem]{ulem}

\usepackage[colorlinks=true, allcolors=blue]{hyperref}
\usepackage[capitalize,nameinlink]{cleveref}

\newcommand{\F}{\mathbb{F}}

\newcommand{\C}{\mathcal{C}}
\newcommand{\AG}{\mathcal{AG}}

\newcommand{\mspan}{\mathrm{span}}

\newtheorem{theorem}{Theorem}[section]
\newtheorem{claim}[theorem]{Claim}
\newtheorem{prop}[theorem]{Proposition}
\newtheorem{lemma}[theorem]{Lemma}
\newtheorem{definition}[theorem]{Definition}

\newcommand{\dist}{\mathop{\rm dist}}
\ifnum\draft=1
\newcommand{\gnote}[1]{{\color{red} [Gabe: #1]}}
\newcommand{\snote}[1]{{\color{blue} [Sumegha: #1]}}
\newcommand{\mnote}[1]{{\color{teal} [Madhu: #1]}}
\else 
\newcommand{\gnote}[1]{}
\newcommand{\snote}[1]{}
\newcommand{\mnote}[1]{}
\fi

\usepackage[utf8]{inputenc}

\newcommand{\Row}{\mathsf{R}}
\newcommand{\Col}{\mathsf{C}}
\renewcommand{\d}{d}

\title{Testing Tensor Products of Algebraic Codes}
\author{Sumegha Garg\thanks{Department of Computer Science, Rutgers University, Piscataway, New Jersey, USA. Part of this work was done when the author was at Harvard University supported by Michael O. Rabin Postdoctoral Fellowship, and at Stanford University supported by the Simons Collaboration on the theory of algorithmic fairness and Simons Foundation Investigators Award 689988. Email: \texttt{sumegha.garg@rutgers.edu}.} \and Madhu Sudan\thanks{School of Engineering and Applied Sciences, Harvard University, Cambridge, Massachusetts, USA. Supported in part by a Simons Investigator Award and NSF Award CCF 2152413. Email: \texttt{madhu@cs.harvard.edu}.} \and Gabriel Wu\thanks{Harvard University, Cambridge, Massachusetts, USA. Email: \texttt{gabrielwu@alumni.harvard.edu}.}}
\date{}

\begin{document}

\maketitle

\begin{abstract}
    Motivated by recent advances in locally testable codes and quantum LDPCs based on robust testability of tensor product codes, we explore the local testability of tensor products of (an abstraction of) algebraic geometry codes. Such codes are parameterized by, in addition to standard parameters such as block length $n$ and dimension $k$, their genus $g$. We show that the tensor product of two algebraic geometry codes is robustly locally testable provided $n = \Omega((k+g)^2)$. Apart from Reed-Solomon codes, this seems to be the first explicit family of two-wise tensor codes of high dual distance that is robustly locally testable by the natural test that measures the expected distance of a random row/column from the underlying code. 
\end{abstract}

\section{Introduction}

In this work we consider the ``robust local testability'' of ``tensor products'' of ``algebraic geometry codes''. We review each of these notions below before describing our results.

All codes considered in this paper are linear codes over some finite field $\F_q$. Given two codes $R \subseteq \F_q^m$ and $C \subseteq \F_q^n$, their tensor product $R \otimes C \subseteq \F_q^{nm}$ consists of all $n \times m$ matrices $M$ whose rows are codewords of $R$ and whose columns are codewords of $C$. It is a simple but valuable exercise in linear algebra to note that the dimension of $R \otimes C$ is the product of the dimensions of $R$ and $C$. It is an arguably simpler exercise to see that the distance of the code $R\otimes C$ is the product of the distance of $R$ and $C$. (In this paper by distance we mean, either absolute or normalized, Hamming distance. This particular statement holds for either of these terms.) Our interest in tensor product codes comes from their potential testability properties elaborated next.

Given an arbitrary matrix $A \in \F_q^{n\times m}$, the definition of a tensor product defines a natural test to see if $A \in R \otimes C$, namely verify every row is in $R$ and every column is in $C$. The ``robust testability'' property explores the robustness of this definition in measuring distances: If the expected distance of a uniformly chosen row of $A$ from $R$ is small, and so is the expected distance of a uniformly chosen column, then is $A$ close to a codeword of $R \otimes C$? A tensor product of codes is said to be robustly testable if the answer is affirmative. (This notion is formulated more precisely and quantitatively in \cref{def:robust}.) 

The formal study of robustness of tensor product codes was initiated by Ben-Sasson and Sudan~\cite{BSS-tensor}. We elaborate on their motivation shortly, but for now mention that most of the general work in this setting, including that in \cite{BSS-tensor} and the work of Viderman~\cite{Viderman:tensor}, considered a variant of the tensor-product testing question raised above. Specifically they consider $m$-wise tensor products of codes for $m \geq 3$ and showed that these are robust with respect to a slightly more complex ``two-wise test'' (where the test measures the expected distance of a random two-dimensional projection from the underlying two-wise tensor) as long as the code being tensored has sufficiently large distance. 

Robustness of two-wise tensors has been shown for very few classes of codes in the literature so far. Briefly the works have considered the tensor product when the codes $R$ and $C$ are Reed-Solomon codes~\cite{RubinfeldS,AroSaf,PolSpi} or when dual distance of $C$ is small~\cite{DSW}\footnote{Dual of a code $C\subseteq \F_q^n$ is defined as the set of vectors orthogonal to $C$, that is, $C^{\perp}=\{x\in\F_q^n\mid \forall y\in C, \langle x,y\rangle=0\}$. By dual distance, we refer to the distance of the dual code. }, or when $R$ and $C$ are random. The work of Valiant~\cite{valiant} also gives examples of asymptotically good codes whose tensor product is not robustly testable --- suggesting that robustness needs additional properties (other than just rate and distance) in the ingredient codes. 

In this work, we consider new classes of codes which are generalizations of Reed-Solomon codes. The codes we study are abstractions of algebraic geometry codes - such codes come as an entire sequence of codes characterized by two main features: (a) they approach the Singleton bound in terms of their distance vs. dimension tradeoff with the additive gap termed the ``genus'' (a parameter that is derived from the genus of some underlying algebraic curves, but gets a purely coding theoretic interpretation in this abstraction); and (b) the Hadamard product of two codewords of two codes of small dimension is a codeword of a code in the sequence of only slightly larger dimension. (See \cref{def:ag-codes} for a precise formulation.) 

The main result in this paper shows that the tensor product of two sequences of AG codes is robustly testable provided their dimension and genus is sufficiently small compared to the length of the code. (See \cref{theorem_robust} for a precise statement.) Before describing getting into the specifics we give some context and motivation for the study of testing of tensor product codes.

\subsection{Robust Testability of Tensor Product: Background and Motivation}

Robust testability of tensor codes was studied in the work of Ben-Sasson and Sudan~\cite{BSS-tensor}, who raised the question of whether $R \otimes C$ is robustly testable for all codes of sufficiently large relative distance. Their question was inspired by the role of the ``bivariate polynomial tester'' in the works on PCPs, originating in the works of Babai, Fortnow, Levin and Szegedy~\cite{BFLS} and explored systematically in the work of Rubinfeld and Sudan~\cite{RubinfeldS}. Seminal results in this space include the work of Arora and Safra~\cite{AroSaf} who showed that Reed-Solomon codes of inverse polynomial rate have constant robustness, and the work of Polishchuk and Spielman~\cite{PolSpi}, who extended the result of \cite{AroSaf} to Reed-Solomon codes of any linear rate bounded away from $1/2$. The question raised in \cite{BSS-tensor} was however answered negatively by Valiant~\cite{valiant} (see also \cite{CopRud,GolMei}, who gave codes of relative distance arbitrarily close to $1$ that were not robustly testable). These works motivated the search for specific classes of codes, other than Reed-Solomon codes, whose tensor product is robustly testable.

The robust testability of tensor product codes has played a central role in many results over the years. The robust testability of Reed-Solomon codes is used ubiquitously in PCP constructions. It also plays a role in the breakthrough result of Ji, Natarajan, Vidick, Wright and Yuen~\cite{MIPRE} showing MIP$^*$=RE. Indeed this motivated the same authors~\cite{quantum-tensor}  to explore the quantum testability of tensor products of codes. It also led to the first combinatorial constructions of LTCs of nearly linear rate in the work of Meir~\cite{Meir} and first known constructions of strong LTCs with nearly linear rate in the work of Viderman~\cite{Viderman}. The robustness of tensor product codes also plays a central role in the recent breakthroughs of Dinur, Evra, Livne, Lubotzky and Mozes~\cite{DELLM} and Panteleev and Kalachev~\cite{PanKal} (see also \cite{LevZem}) giving constant query LTCs of linear rate and distance, and quantum LDPCs. We remark that the role of tensor product codes here is essential --- these are the only sources of ``redundancy'' among the local constraints in the constructions of \cite{DELLM,PanKal,LevZem} and such redundancies are necessary to get LTCS as shown in \cite{BGKSV}. Indeed the only two known sources of redundancy among constraints in LDPCs come from symmetries (see \cite{KauSud}) or the tensor product construction. 

These many applications motivate the quest for general tools in the analysis of robustness of tensor product codes. To date the only known works on robust testability of tensor codes are (1) the aforementioned works on Reed-Solomon codes, (2) a work of Dinur, Sudan and Wigderson~\cite{DSW} roughly showing that $R \otimes C$ is robustly testable if $C$ is an LDPC code and (3) recent works, notably by Panteleev and Kalachev~\cite{PanKal} and Leverrier and Z\'emor~\cite{LevZem}, showing that tensor products of random linear codes (and their duals) are robustly testable\footnote{In this work, we prove robust testability of tensor products of an abstraction of algebraic geometry codes. As the dual distance of LDPC codes is a constant by definition, this gives us first explicit family of codes of super-constant dual distance that is robustly locally testable after Reed-Solomon codes (to the best of our knowledge). Explicit constructions for AG codes  with super-constant dual distance are known (see Section 2.7 in \cite{hoholdt1998algebraic}). }, with improved parameters in \cite{KalPan:2-local,DinurHLV23}. (We remark that some of the interest in 2-dimensional robust testability is due to an equivalence with a notion called  ``product expansion'' of codes, an equivalence that holds only for two-dimensional tensor products. Both the notion of robustness and product expansion do extend to higher dimensional products but they are no longer equivalent. See~\cite{KalPan:3-local} and references therein.)


The lack of more general results motivates us to study the testability results for products of Reed-Solomon codes. (We note the other explicit result~\cite{DSW} is already generic, while the robustness in the setting of Reed-Solomon codes is not.) However, the current analyses of robustness in this setting are very algebraic and to make this analysis more generic, one needs an appropriate abstraction of the underlying algebra and we discuss this next.

\subsection{Abstracting Algebraic Codes} 

Over the years coding theorists have proposed nice abstractions of algebraic codes - see for instance~\cite{DuuKot,Pellikaan}. These works abstract a product property that captures the fact that the product of two polynomials of degree $d$ has degree at most $2d$. 
This corresponds to the fact that the coordinate-wise product of two codewords from a (roughly) $d$ dimensional space is contained in a (roughly) $2d$ dimensional space, which is a highly non-trivial effect. (In contrast for a generic linear code $C \subseteq \F^n$ of dimension $d$, the smallest linear space that contains all the coordinate-wise products of pairs of codewords from $C$ has dimension $\min\{d^2,n\}$.) This non-trivial product property seen in Reed-Solomon codes when abstracted properly captures most algebraic codes (including Reed-Muller and algebraic geometry codes) nicely and suffices to explain most decoding algorithms for such codes. However other algorithms, such as list-decoding algorithms, use more involved properties (see \cite{GuSu-ag}) that include the ability to capture multiplicities of zeroes and the ability to shift a polynomial without increasing its degree (i.e., $f(x)$ has the same degree as $f(x+a)$). 

The current analyses of robust testability of tensor products of Reed-Solomon codes use many aspects of polynomials in addition to the product property. For instance they rely on unique factorization, on the role of resultants in computing greatest common divisors (and the fact that resultants themselves are low-degree polynomials). They use the fact that puncturing of Reed-Solomon codes are Reed-Solomon codes etc. Given all these aspects in the proofs, it is interesting to see how far one can go with more minimal assumptions. 

In this work, we use a simple quantitative version of the product property (see \cref{def:ag-codes}) which naturally captures algebraic geometry codes, (but not for instance Reed-Muller codes). In particular we avoid use of unique factorization and GCDs, and also avoid explicit use of the puncturing property. This allows us to recover a moderately strong version of the analysis for Reed-Solomon codes: specifically we can analyze codes that have block length at least quadratic in the dimension. Thus our work is not strong enough to imply the result of Polishchuk and Spielman~\cite{PolSpi} who only require block length linear in the dimension; but is stronger than the previous work of Arora and Safra~\cite{AroSaf} (and implies their result) who showed that the tensor product of two Reed-Solomon codes of dimension $d$ and length $n$ is robustly testable provided $n = \Omega(d^3)$. 

\mnote{Can try adding an overview of proof here if we have time.}

The next section presents our formal definitions and theorem statement.

\section{Definitions and Main Result}

We use $\F_q$ to denote the finite field on $q$ elements. 
We use functional notation to describe vectors, so the vector space $\F_q^n$ will be viewed and represented as functions $\F_q^n = \{f:S\to \F_q\}$ for some set $S$ with $|S| = n$. (Often we use $S = [n]$.) Note that with this notation we naturally have the notion of $f+g$ and $fg$ both of which are in $\F_q^n$. For functions $f:S \to \F_q$ and $g:T \to \F_q$, we use $f \otimes g: S \times T \to \F_q$ to denote the function $(f\otimes g)(x,y) = f(x)g(y)$. If $f\in \F_q^m$ and $g \in \F_q^n$, note that we have $f \otimes g \in \F_q^{n \times m} \cong \F_q^{nm}$. Here, we assume that a ``row" is indexed by $y$ and varies $x\in [m]$ (vice versa for ``column").
For $f,g:S \to \F_q$,  we use $\dist(f,g)$ to denote the absolute (non-normalized) Hamming distance between $f$ and $g$, i.e., $\dist(f,g) = |\{x \in S | f(x) \ne g(x) \}|$ and we use $\delta(f,g) = \frac{1}{|S|} \cdot \dist(f,g)$ to denote the normalized Hamming distance.

We consider linear codes $C \subseteq \F_q^n$. For such a code we use $\dim(C)$ to denote its dimension as a vector space, $\dist(C)$ to denote its (non-normalized) minimum distance (between any two code vectors). For a vector $f \in \F_q^n$ and code $C \subseteq \F_q^n$ we use $\delta(f,C)$ to denote the distance of $f$ to the nearest codeword in $C$, i.e., $\min_{g \in C} \{\delta(f,g)\}$. 
For codes $C_1, C_2 \subseteq \F_q^n$, we use $C_1 \star C_2$ to denote its Hadamard product, i.e., $C_1 \star C_2 = \mspan(\{fg | f \in C_1, g \in C_2\})$. For codes $C_1 \subseteq \F_q^m$ and $C_2 \subseteq \F_q^n$, we let $C_1 \otimes C_2$ denote their tensor product, i.e.,  $C_1 \otimes C_2 = \mspan(\{f \otimes g | f \in C_1, g \in C_2\})$. For a matrix $A\in \F_q^{nm}$, it is a simple exercise to see that $A\in C_1 \otimes C_2$ iff every row of $A$ is in $C_1$ and every column is in $C_2$ (when $C_1$ and $C_2$ are linear codes).

\begin{definition}[(Abstract) Algebraic Geometry Code]\label{def:ag-codes} Given integers $0\leq g \leq n$ and a prime power $q$, an $(q,n,g)$-algebraic geometry code sequence is a sequence of linear codes $\C = \{\C(\ell) \subseteq \F_q^n \mid 0 \leq \ell \leq n\}$ with the following requirements: 

(1) For every $\ell$, $\dim(\C(\ell)) \geq \ell-g$ and $\dist(\C(\ell)) \geq n-\ell$. 

(2) For every $\ell, m$, $\C(\ell) \star \C(m) \subseteq \C(\ell+m)$.

We use $\AG(q,n,g)$ to denote the space of all $(q,n,g)$-algebraic geometry code sequences.
\end{definition}

Condition (2) above is called the product property and it is what makes algebraic codes special.

\begin{definition}[Robustness of Tensor Product]\label{def:robust} For codes $C_1 \subseteq \F_q^m$ and $C_2 \subseteq \F_q^n$ and $0 \leq \rho \leq 1$ we say that $(C_1,C_2)$ is $\rho$-robust
if for every $F\in \F_q^{n \times m}$ we have 
$$ \rho \cdot \delta(F,C_1 \otimes C_2) \leq \frac12[\delta(F,C_1 \otimes \F_q^n) + \delta(F,\F_q^m \otimes C_2)].$$
\end{definition}

For a linear code $C$ and set $A \subseteq S$, let $C|_A = \{x|_{A} \,\mid  x \in C\} \subseteq \F_q^{|A|}$ be the projection of $C$ to the coordinates of $A$. We will use the following proposition about tensor products in our proof. 
\begin{prop}
\label{tensor_restrict}
    If $C_1, C_2 \subseteq \F_q^n$ are linear codes and $A, B \subseteq [n]$, then 
    \[C_1|_A \otimes C_2|_B = (C_1 \otimes C_2)|_{A \times B}.\]
\end{prop}

\proof
    Let $\dim C_1 = k_1$ and $\dim C_2 = k_2$. Then there exist matrices $M_1 \in \F_q^{k_1 \times n}$ and $M_2 \in \F_q^{k_2 \times n}$ such that
    $$C_1 = \{xM_1 \,|\, x \in \F_q^{k_1}\} \;\;\text{ and }\;\; C_2 = \{xM_2 \,|\, x \in \F_q^{k_2}\},$$
  and their tensor product can then be expressed as
    $$C_1 \otimes C_2 = \{M_2^T X M_1 \,|\, X \in \F_q^{k_2 \times k_1}\}.$$

    Note that $C_1|_A$ is generated by the restriction of $M_1$ to the columns of $A$ (i.e. $M_1 |_{[k_1] \times A}$). Similarly for $C_2|_B$. Our goal is now to show that
    $$\{(M_2^T|_{B \times [k_2]}) X (M_1|_{[k_1] \times A}) \,|\, X \in \F_q^{k_2 \times k_1}\} = \{(M_2^T X M_1)|_{B \times A} \,|\, X \in \F_q^{k_2 \times k_1}\}.$$
    It suffices to show that $\forall X$,
    $$(M_2^T|_{B \times [k_2]}) X (M_1|_{[k_1] \times A}) = (M_2^T X M_1)|_{B \times A}.$$
    This follows because:
    \begin{align*}
    (M_2^T|_{B \times [k_2]}) X (M_1|_{[k_1] \times A}) &= (M_2^T X)|_{B \times [k_1]} (M_1|_{[k_1] \times A}) \\
    &= (M_2^T X (M_1|_{[k_1] \times A}))|_{B \times A}  \\
    &= (M_2^T (X M_1)|_{[k_2] \times A}))|_{B \times A}  \\
    &= ((M_2^T X M_1)|_{[n] \times A})|_{B \times A}  \\
    &= (M_2^T X M_1)|_{B \times A}.&\qedhere
    \end{align*}

\begin{theorem}
\label{theorem_robust}
    There exist constants $\rho > 0$ and $c_0 < \infty$ such that for every triple of non-negative integers $n,\ell,g$ and prime power $q$ satisfying $\ell > \max\{c_0, g\}$ and $n > c_0(\ell+g)^2$ we have the following: If $\C_1,\C_2 \in \AG(q,n,g)$, then $(\C_1(\ell),\C_2(\ell))$ is $\rho$-robust. 
\end{theorem}

\cref{theorem_robust} is proved at the end of \cref{sec:proof-main}. Our proof combines elements from the proofs of \cite{AroSaf} and \cite{PolSpi}. A direct use of the proof from \cite{AroSaf} would have resulted in a $n = \Omega((k+g)^3)$ requirement. On the other hand, the proof of \cite{PolSpi} uses properties of unique factorization domain, which are no longer true in this general setting. By combining different ingredients we are able to get a bound that is intermediate while still being general. 

The fact that $n = \Omega(g^2)$ for the theorem to be useful does limit its applicability. Nevertheless the theorem is not vacuous even given the testability of Reed-Solomon codes and in particular, there exist infinitely many AG codes\footnote{For example, for any prime $p$, there is an elliptic curve over $\F_p$ with $p+1$ number of rational points (Theorem 14.18 in \cite{cox2022primes}). One can apply the construction of AG codes from Chapter 2 in \cite{stichtenoth2009algebraic} to these elliptic curves to get  $(p,p,1)$-sequences of AG codes.} with $g = o(\sqrt{n})$ and $n \approx q$.

\section{Proof of {\protect \cref{theorem_robust}}}\label{sec:proof-main}


Let $\C_1,\C_2$ be a pair of $(q,n,g)$-sequences of algebraic geometry codes. We first show that for a tensor product of codes $\C_1(\ell)$ and $\C_2(\ell)$, if the rows of a matrix are near codewords of $\C_1(\ell)$ and columns are near codewords of $\C_2(\ell)$, then it is also close to a codeword of the tensor product code.
\begin{theorem}
\label{theorem_main}
There exist positive constants $\epsilon_0, c_0$ that make the following true.

For all $0 < \epsilon < \epsilon_0$, integers $n, g, \ell$ and $q$ such that $\ell > \max\{c_0, g\}$, and $n > c_0 (\ell + g)^2$, and for all $(q,n,g)$-sequences of AG codes $\C_1,\C_2$: If $R\in \C_1(\ell) \otimes \F_q^n$ and $C \in \F_q^n \otimes \C_2(\ell)$ are such that $\delta(R,C)=\epsilon$, then there exists $Q\in \C_1(\ell)\otimes \C_2(\ell)$ such that 
\[\delta(Q,R) + \delta(Q, C) \le 2\epsilon.\]
\end{theorem}

\begin{proof}
We will prove the theorem for $\epsilon_0 = \frac{1}{100}$ and $c_0 = 15$.
Define the following constants:
$$\gamma = 2\sqrt{\epsilon}$$
$$\gamma' = 2\gamma = 4\sqrt{\epsilon}$$
$$L = 2 (\ell + g)$$
$$\d = \lfloor \sqrt{\epsilon} L \rfloor + g + 2$$
Note that these constants are chosen to satisfy the following inequalities:

\begin{equation} \label{ineq1}
    \gamma < \gamma'(1-\gamma')
\end{equation}
\begin{equation} \label{ineq2}
    \left(1-\frac{\epsilon}{\gamma^2}\right)L > d + \ell
\end{equation}
\begin{equation} \label{ineq3}
    n(1 - \gamma - \gamma') - dL > L
\end{equation}
\begin{equation} \label{ineq4}
    (2\sqrt{\epsilon} + \epsilon) \cdot \frac{n}{n-\ell} < 3\sqrt{\epsilon}
\end{equation}
\begin{equation} \label{ineq5}
    \frac{n - \ell - 3\sqrt{\epsilon}n}{n} > \frac 12
\end{equation}

We can quickly check that all of these inequalities hold for our choice of $\epsilon_0$ and $c_0$. Inequality \eqref{ineq1} holds because $1 < 2(1-4\sqrt{\epsilon})$. Inequality \eqref{ineq2} holds because the left-hand side can be written as $1.5(\ell + g)$ and the right-hand side can be written as $\lfloor \sqrt{\epsilon} \cdot 2(\ell+g) \rfloor + g + 2 +\ell$, which is at most $(1+2\sqrt{\epsilon})(\ell+g) + 2$.  Inequality \eqref{ineq3} holds because $(1+d)L < (1 + 2\sqrt{\epsilon}) (\ell + g)  \cdot 2(\ell + g) < 6(\ell + g)^2 < 6n/c_0  \leq n(1 - \gamma - \gamma').$ Inequality \eqref{ineq4} holds because $\frac{n}{n-\ell} < \frac{3}{2.1}$ and $2\sqrt{\epsilon} + \epsilon < 2.1 \sqrt{\epsilon}$. Inequality \eqref{ineq5} holds because $\ell < \frac{1}{15}n < (\frac 12 - 3\sqrt{\epsilon_0})n \implies \ell + 3\sqrt{\epsilon}n < \frac{n}{2}$.
Throughout the proof, we will note where we make use of each inequality.

Define an ``error'' to be a location $(x, y)$ at which $R(x, y) \neq C(x, y)$. The ``error fraction'' of a row is the number of errors in that row divided by its length (and similarly for columns). Note that the length of a row isn't always $n$; sometimes we will use notions like ``the error fraction of a row among the first $L$ columns''.

The majority of the proof of Theorem \ref{theorem_main} is spent on proving the following Lemma.

\begin{lemma}
\label{lemma1}
There exists a $Q \in \C_1(\ell) \otimes \C_2(\ell)$ such that $$\delta(Q, R) \leq 2 \sqrt{\epsilon}.$$
\end{lemma}

\begin{proof}
First, remove any row or column that has an error fraction greater than $\epsilon/\gamma'$. This removes less than $\gamma'$ fraction of the rows and columns (as $\delta(R,C)=\epsilon$). Call the sets of these deleted rows and columns $\Row_1, \Col_1$, respectively. After this, it is guaranteed that each row and column has at most $n\epsilon/\gamma'$ errors, and its new length is at least $n(1-\gamma')$, so it has an error fraction at most $\epsilon/(\gamma'(1-\gamma')) < \epsilon/\gamma$ by Inequality \eqref{ineq1}. The total error fraction of the matrix is still at most $\epsilon$ because the removed columns and rows each had error fractions above $2\epsilon$, so overall the removed entries had an error fraction above $\epsilon$.

Next, choose a submatrix $M \subset ([n] \setminus \Row_1) \times ([n] \setminus \Col_1)$ of size $L \times L$ such that the error fraction within $M$ is at most $\epsilon$ (in fact, the number of errors in $M$ is at most $\lfloor L^2 \epsilon\rfloor$). Such an $M$ must exist because choosing a random submatrix results in an error fraction at most $\epsilon$ in expectation. WLOG, say that $M$ lies in the top left of the matrix (so its rows and columns are indexed by $[L]$).

\begin{claim}
\label{claim1}
Recall that $\d = \lfloor\sqrt{\epsilon}L\rfloor + 2 + g$. There exist vectors $0 \neq E\in \C_1(\d)\otimes \C_2(\d)$ and $N\in \C_1(\d + \ell)\otimes \C_2(\d + \ell)$ such that 
\[E(x,y)R(x,y)=E(x,y)C(x,y)=N(x,y),\;\; \forall (x,y)\in M.\]
\end{claim}

\begin{proof}
Consider the projection of $\C_1(\d)\otimes \C_2(\d)$ onto the set of coordinates $\{(x,y) \in M \,|\, R(x,y) \neq C(x,y)\}$. Since $\C_1(\d)\otimes \C_2(\d)$ has dimension at least $(\d-g)^2 = (\lfloor \sqrt{\epsilon}L\rfloor + 2)^2 > \epsilon L^2$, but $\{(x,y) \in M \,|\, R(x,y) \neq C(x,y)\}$ has size at most $\epsilon L^2$, the projection has a non-trivial kernel. Choose any nonzero $E$ in this kernel. 

We have ensured that $E(x,y) = 0$ whenever $R$ and $C$ disagree on submatrix $M$, so
$$E(x,y)R(x,y) = E(x,y)C(x,y) \quad \forall (x,y) \in M.$$

Let's look at the projection of $ER$ on the submatrix $M$. Every row is an element of $\C_1(\d + \ell)|_{[L]}$ and every column is an element of $\C_2(\d + \ell)|_{[L]}$ (by product property of AG codes). Thus, $ER|_M$ is an element of $\C_1(\d + \ell)|_{[L]} \otimes \C_2(\d + \ell)|_{[L]}$. By Proposition \ref{tensor_restrict}, any element of $\C_1(\d + \ell)|_{[L]} \otimes \C_2(\d + \ell)|_{[L]} = (\C_1(\d + \ell) \otimes \C_2(\d + \ell))|_M$ can be extended to an element of $\C_1(\d + \ell) \otimes \C_2(\d + \ell)$. Choose any such extension and call it $N$. So we have
$$E(x,y)R(x,y) = E(x,y)C(x,y) = N(x,y) \quad \forall (x,y) \in M,$$
as desired. This proves Claim \ref{claim1}.
\end{proof}
Next we extend the relation between $E$, $R$, $C$ and $N$ to almost the entire matrix.

\begin{claim}
\label{claim2}
$E(x,y)R(x,y)=E(x,y)C(x,y)=N(x,y)$ for all  $(x,y)$ but $\gamma$ fraction of the remaining rows and columns. More formally, there exist sets $\Row_2 \subset [n] \setminus \Row_1$ and $\Col_2 \subset [n] \setminus \Col_1$ such that the equality holds on all of $([n] \setminus (\Row_1 \cup \Row_2)) \times ([n] \setminus (\Col_1 \cup \Col_2))$, and  that $$|\Row_2| \leq \gamma (n- |\Row_1|) \text{ and } |\Col_2| \leq \gamma (n- |\Col_1|).$$
\end{claim}

\begin{proof}
Call a row $y \in [n] \setminus \Row_1$ ``bad'' if

$$\Pr_{x \sim [L]}[R(x,y) \neq C(x,y)] > \frac{\epsilon}{\gamma}\cdot \frac{1}{\gamma}.$$

In other words, a bad row is one that has too many errors in the columns of $M$. Define a bad column analogously. Let the sets of bad rows and columns be $\Row_2$ and $\Col_2$, respectively.
Call the remaining rows and columns ``good'' (i.e. the good rows are $[n] \setminus (\Row_1 \cup \Row_2)$). The fraction $|\Row_2|/(n - |\Row_1|)$ must be at most $\gamma$ because, among the first $L$ columns, there is a combined error fraction of at most $\epsilon/\gamma$ (recall that each such column has at most $\epsilon/\gamma$ error fraction). The same holds for the columns.

Now we prove that $E(x,y) \cdot C(x, y) = N(x, y)$ on all good columns. Fix any good column $\hat x$. For any row $y_0 \in [L]$, since  row vectors $E (x, y_0) \cdot R(x, y_0)$ and $N(x, y_0)$ belong to $\C_1(\d + \ell)$ and agree on $L$ points (i.e. the columns in $[L]$), they must be identical because the distance of the code $C_1(\d + \ell)$ is at least  $n-\d - \ell$ which is greater than $n-L$ (by Inequality \eqref{ineq2}). So $E(\hat x, y_0) \cdot R(\hat x, y_0) = N(\hat x, y_0)$. Since $\hat x$ is good, there are at least $L\left(1-\epsilon/\gamma^2\right)$ values of $y_0 \in [L]$ such that

$$E(\hat x, y_0) \cdot C(\hat x, y_0) = E(\hat x, y_0) \cdot R(\hat x, y_0) = N(\hat x, y_0).$$

Now, since
$$L\left(1-\frac{\epsilon}{\gamma^2}\right) > \d + \ell, $$
by Inequality \eqref{ineq2}, this tells us that  the column vectors $E(\hat x, y) \cdot C(\hat x, y)$ and $N(\hat x, y)$ (which belong in $\C_2(\d + \ell)$) are identical. Thus, $E(x,y) \cdot C(x,y) = N(x,y)$ on all good columns. The same argument can be applied on the rows to show that $E(x,y) \cdot R(x,y) = N(x,y)$ on all good rows. On the intersection of all good columns and rows, we have
$$E(x,y)\cdot R(x,y)=N(x,y)=E(x,y)\cdot C(x,y).$$
This proves Claim \ref{claim2}.
\end{proof}

Now, we have the necessary claims to prove Lemma \ref{lemma1}. Our first step is to find a submatrix of size $L \times L$ upon which $E(x, y)$ is never $0$. To do this, consider the good submatrix $G = ([n] \setminus (\Row_1 \cup \Row_2)) \times ([n] \setminus (\Col_1 \cup \Col_2))$ upon which $E(x,y) \cdot R(x,y) = E(x,y) \cdot C(x,y)$ from Claim \ref{claim2}. Note that $G$ has at least $n(1-\gamma' - \gamma)$ rows and columns. Choose $L$ columns $\{x_1, \dots, x_L\} \subset [n] \setminus (\Col_1 \cup \Col_2)$ upon which $E$ is not identically $0$ (these must exist: simply consider a row upon which $E$ is not identically $0$, then choose among the $n(1-\gamma'-\gamma) - \d > L$ columns in $[n] \setminus (\Col_1 \cup \Col_2)$ in which the row takes a nonzero value). 

In each of these columns, $E$ is $0$ at most $\d$ times (this follows by the distance of the code; the all-zero vector is always a codeword of a linear code). Thus, there must be $n(1 - \gamma' - \gamma) - dL$ rows in $[n] \setminus (\Row_1 \cup \Row_2)$ upon which $E$ is never $0$ in the columns $\{x_1, \dots, x_L\}$. But since $n(1-\gamma' - \gamma) - dL > L$ by Inequality \eqref{ineq3}, we can find an $L \times L$ submatrix of $G$ upon which $E$ is never $0$. Call this submatrix $M'$. Since $E(x,y) \cdot R(x,y) = E(x,y) \cdot C(x,y)$ on $ M'\subset G$, we know that $R(x,y) = C(x,y)$ on $M'$.

Permute the matrix so that the rows and columns of $M'$ are indexed by $[L]$. Every row of $M'$ is an element of $\C_1(\ell)|_{[L]}$ and every column of $M'$ is an element of $\C_2(\ell)|_{[L]}$. Thus, by Proposition \ref{tensor_restrict}, the submatrix $M'$ is an element of $\C_1(\ell)|_{[L]} \otimes \C_1(\ell)|_{[L]} = (\C_1(\ell) \otimes \C_2(\ell))|_{M'}$ and can be extended to some $Q \in \C_1(\ell) \otimes \C_2(\ell)$ on the entire matrix, such that $R(x, y) = C(x, y) = Q(x, y)$ on $M'$. 

For this last step, we can bring back all of the rows and columns $\Row_1, \Col_1, \Row_2, \Col_2$ that were previously ignored. Let $\Row_3 \subset [n]$ be the set of all rows that have an error fraction above $\epsilon/\gamma^2$ among the first $L$ columns. Since each column in $[L] \subset [n] \setminus \Col_1$ has an overall error fraction at most $\epsilon/\gamma$, we know that $|\Row_3| < \gamma n$.

We will show that $R$ is identical to $Q$ on all rows not in $\Row_3$. Consider any row $\hat y \in [n] \setminus \Row_3$. For any column $x_0 \in [L]$, we know that the column vectors $Q(x_0, y)$ and $C(x_0, y)$ belong to $\C_2(\ell)$ and are identical (using distance of the code), so $Q(x_0, \hat y) = C(x_0, \hat y)$. Thus, for any $x_0$ such that $R(x_0, \hat y) = C(x_0, \hat y)$, we have $Q(x_0, \hat y) = R(x_0, \hat y)$ as well. Since we're assuming that there is an error fraction of at most $\epsilon/\gamma^2$ among the first $L$ columns of row $\hat y$,  row vectors $Q(x, \hat y)$ and $R(x, \hat y)$ (in $\C_1(\ell)$) are equal in at least $(1-\epsilon/\gamma^2)L$ places. Thus, $Q(x, \hat y)$ is identical to $R(x, \hat y)$ (using Inequality \eqref{ineq2} and the fact that distance of code $\C_1(\ell)$ is at least $n-\ell$).

Finally, this means that $Q(x,y) = R(x,y)$ on all points in $([n] \setminus \Row_3) \times [n]$. Since $\Row_3$ contains at most $\gamma$ fraction of all rows, we know $\delta(Q, R) \leq \gamma = 2\sqrt{\epsilon}$. This proves Lemma \ref{lemma1}.
\end{proof}

Now we prove Theorem \ref{theorem_main}.
If two $\C_1(\ell)$ row vectors are distinct, they must disagree on at least $n - \ell$ points (which follows from the distance of the code). Thus, the fraction of rows upon which $R$ and $Q$ disagree anywhere is at most $\frac{n}{n - \ell}2 \sqrt{\epsilon} < 3 \sqrt{\epsilon}$ by Inequality \eqref{ineq4}. Similarly, the fraction of errors between $C$ and $Q$ is at most $2\sqrt{\epsilon} + \epsilon$, so the fraction of columns upon which $C$ and $Q$ disagree is at most $\frac{n}{n - \ell}(2\sqrt{\epsilon} + \epsilon) < 3 \sqrt{\epsilon}$, again by Inequality \eqref{ineq4}. Permute the matrix so that these rows and columns are contiguous in the bottom right. Label the four regions of the matrix $A_{11}$, $A_{12}$, $A_{21}$, $A_{22}$:

\begin{center}
\begin{tikzpicture}
\draw[thick] (0,0) rectangle (5,5);
\draw[thick] (0,0) rectangle (4,1);
\draw[thick] (4,1) rectangle (5,5);
\node at (2,3) {$A_{11}$};
\node at (4.5,3) {$A_{12}$};
\node at (2,.5) {$A_{21}$};
\node at (4.5,.5) {$A_{22}$};
\end{tikzpicture}
\end{center}

where $A_{11} \cup A_{12}$ is the region where $R(x,y) = Q(x,y)$ and $A_{11} \cup A_{21}$ is the region where $C(x,y) = Q(x,y)$. The submatrix $A_{22}$ has size at most $3\sqrt{\epsilon}n$ in each dimension. Now, notice that each row of $A_{21}$ has at least $n - \ell - 3\sqrt{\epsilon}n$ disagreements between $R$ and $Q = C$. Thus,

$$\frac{\dist_{A_{21}}(R, C)}{|A_{21} \cup A_{22}|} \geq \frac{n - \ell - 3\sqrt{\epsilon}n}{n} > \frac 12$$
by Inequality \eqref{ineq5}, where $\dist_M(R, C)$ is the number of disagreements between $R$ and $C$ on submatrix $M$.
Applying the same logic to the columns, we get 
$$\frac{\dist_{A_{12}}(R, C)}{|A_{12} \cup A_{22}|} > \frac 12.$$
But now, since
$$\dist_{A_{12}}(R,C) + \dist_{A_{21}}(R,C) \leq \dist(R,C) = \epsilon n^2,$$
we can combine these inequalities to get
$$|A_{12} \cup A_{22}| + |A_{21} \cup A_{22}| < 2\epsilon n^2.$$
Finally, since $\dist(Q, R) \leq |A_{12} \cup A_{22}|$ and $\dist(Q, C) \leq |A_{21} \cup A_{22}|$, we conclude that
$$\delta(Q, R) + \delta(Q, C) \leq 2\epsilon.$$
This finishes the proof of Theorem \ref{theorem_main}.
\end{proof}

\begin{proof}[\textbf{Proof of Theorem \ref{theorem_robust}}]
We prove the theorem for $\rho = \epsilon_0/2$. The constants $c_0$ and $\epsilon_0$ are defined in Theorem \ref{theorem_main}. For any $F \in \F_q^{n \times n}$, let $R \in \C_1(\ell) \otimes \F_q^n$ and $C \in \F_q^n \otimes \C_2(\ell)$ be the closest vectors to $F$ in their respective codes. Let $\epsilon = \delta(R, C)$.

If $\epsilon \geq \epsilon_0$, then $\rho = \epsilon_0/2$ trivially works because
$$\frac{\epsilon_0}{2} \cdot \delta(F, \C_1(\ell) \otimes \C_2(\ell)) \leq \frac{\epsilon_0}{2} \leq \frac 12 \delta(R, C) \leq \frac 12 [\delta(F, R) + \delta(F, C)].$$

Otherwise, if $\epsilon < \epsilon_0$, we apply Theorem \ref{theorem_main} to find a $Q \in \C_1(\ell) \otimes \C_2(\ell)$ such that $\delta(Q, R) + \delta(Q, C) \leq 2 \epsilon$. Then we know
\begin{align*}
    \delta(F, Q) &\leq \min(\delta(F, R) + \delta(R, Q), \delta(F, C) + \delta(C, Q)) \\
    &\leq \frac 12 \left(\delta(F, R) + \delta(F, C) + \delta(R, Q) + \delta(C, Q)\right)\\
    &\leq \frac 12 \left(\delta(F, R) + \delta(F, C) + 2\epsilon\right)\\
    &\leq \frac 32 \left(\delta(F, R) + \delta(F, C)\right) \\
    \frac 13 \delta(F, Q) &\leq \frac 12 \left[\delta(F, R) + \delta(F, C)\right].
\end{align*}
Since $\rho < 1/3$, the theorem is proven.
\end{proof}
\section*{Acknowledgements}
We would like to thank Venkatesan Guruswami and Swastik Kopparty for valuable conversations on AG codes, and the anonymous referees of Random 2025 for their careful reading and corrections.

\bibliographystyle{alpha}
\bibliography{refs}

\end{document}